\documentclass[]{article}
\usepackage{proceed2e}
\usepackage{times}

\usepackage{mystyle}
\usepackage{xcolor}

\begin{document}

\title{Revisiting Non-Progressive Influence Models: Scalable Influence Maximization in Social Networks}
\author{ {\bf Golshan Golnari \thanks{$^*$ G. Golnari and A. Asiaee contributed equally to this work.}} \\
	Electrical and Computer \\
	Engineering Dept. \\
	University of Minnesota\\
	Minneapolis, MN 55455 \\
	\And
	{\bf Amir Asiaee T. $^*$}  \\
	Computer Science and\\
	Engineering Dept. \\
	University of Minnesota\\
	Minneapolis, MN 55455 \\
	\And
	{\bf Arindam Banerjee}   \\
	Computer Science and\\
	Engineering Dept. \\
	University of Minnesota\\
	Minneapolis, MN 55455 \\
	\And
	{\bf Zhi-Li Zhang}   \\
	Computer Science and\\
	Engineering Dept. \\
	University of Minnesota\\
	Minneapolis, MN 55455 \\
}

\maketitle

\begin{abstract} 
\vspace*{-0.05in}
Influence maximization in social networks has been studied extensively in computer science community for the last decade. 
However, almost all of the efforts have been focused on the \emph{progressive} influence models, such as independent cascade (IC) and Linear threshold (LT) models, which cannot capture the \textit{reversibility of choices}. 
In this paper, we present the Heat Conduction (HC) model which is a \textit{non-progressive influence model} and has favorable real-world interpretations. 
Moreover, we show that HC unifies, generalizes, and extends the existing non-progressive models, such as the Voter model \cite{even-dar_note_2007} and non-progressive LT \cite{kempe_maximizing_2003}. 
We  then prove that selecting the optimal seed set of influential nodes is NP-hard for HC but by establishing the submodularity of influence spread, we can tackle the influence maximization problem with a scalable and provably near-optimal greedy algorithm. 
To  the  best  of  our  knowledge, we  are  the  first  to present a scalable solution for influence maximization under non-progressive LT model, as a special case of HC model. 
In sharp contrast to the other greedy influence maximization methods, our fast and efficient \textsc{C2Greedy} algorithm benefits from two analytically computable steps: \emph{closed-form} computation for finding the influence spread as well as the greedy seed selection. 
Through extensive experiments on several and large real and synthetic networks, we show that \textsc{C2Greedy} outperforms the state-of-the-art methods, under HC model, in terms of both influence spread and scalability. 

\end{abstract}
\vspace*{-0.05in}

\vspace*{-2pt}
\section{INTRODUCTION}\label{sec:intro}
\vspace*{-0.05in}
 
Motivated by viral marketing and other applications, the problem of  influence maximization in a social network has attracted much attention in recent years. Given a social network where nodes represent users in a social group, and edges represent relationships and interactions between the users (and through which they influence each other), the basic idea of influence maximization is to select an initial set of ``most influential'' users (often referred to as the {\em seeds}) among all users  so as to maximize the total influence under a given diffusion process (often referred to as the {\em influence model}) on the social network. In the context of viral marketing, this amounts to by initially targeting a set of influential customers,  e.g., by providing them with free product samples, with the goal to trigger a cascade of influence through ``word-of-mouth'' or recommendations to friends to maximize the total number of customers adopting the said product.
Domingos and Richardson \cite{domingos_mining_2001} introduced this algorithmic problem to the Computer Science community and Kempe et al. \cite{kempe_maximizing_2003} made the topic vastly popular under the name of \emph{influence maximization}.
They studied two influence models, the independent cascade (IC) model and the linear threshold (LT) model, and applied a greedy method to tackle the influence maximization problem \cite{kempe_maximizing_2003}.
Unfortunately Kempe et al.'s approach \cite{kempe_maximizing_2003} for calculating the influence spread is based on Monte Carlo simulations which does not scale to large networks \cite{chen_scalable_2010, chen_scalable_20101}. As the result, it motivated researchers to either improve the scalability \cite{chen_scalable_2010, chen_scalable_20101} or study more tractable influence models \cite{rodriguez_uncovering_2011, due_scalable_2013}. 

The focus of almost all of these earlier studies are, however, \emph{progressive} influence models, including LT and IC models, in which once a costumer adopts a product or a user performs an action she cannot revert it. Retweeting news and 
sharing videos in online social network websites, 
are examples of progressive, i.e. irreversible actions. 
Nevertheless, there are numerous real world instances where the actions are \emph{non-progressive} especially in technology adoption domain. 
For example, adopting a cell phone service provider, such as AT\&T and T-mobile, is a non-progressive action where a user can switch between providers. The objective of influence maximization in this example is to persuade more users to adopt the intended provider for a longer period of time. 
To capture the reversibility of choices in real scenarios, we present Heat Conduction (HC) model which has favorable real-world interpretation. We also show that HC unifies, generalizes, and extends the existing non-progressive models, including non-progressive LT (NLT) \cite{kempe_maximizing_2003} and Voter model \cite{even-dar_note_2007} (see Section \ref{sec:discussion}). In contrast to the Voter model, HC does not \textit{necessarily} reach consensus, where one product dominates and extinguishes the others after finite time, so the proposed HC model can explain the \emph{coexistence of multiple product adoptions}, which is a typical phenomena in real world. In addition, HC model incorporates both ``social'' and ``non-social''  factors, e.g., intrinsic inertia or reluctance of some users in adopting a new idea or trying out a new product, external ``media effect'' which exerts a ``non-social'' influence in promoting certain ideas or products. 

We tackle the influence maximization problem under HC influence model with a \emph{scalable} and provably \emph{near-optimal} solution.
Kempe et al.'s approach \cite{kempe_maximizing_2003} for influence maximization under NLT model, is to reduce the model to (progressive) LT by replicating the network as many as time progresses and compute the influence spread by the same slow Monte Carlo method for the resulted huge network. This approach is practically impossible for large networks, specially for the \emph{infinite time horizon}. 
We also prove that contrary to the Voter, for which the influence maximization can be solved \emph{exactly} in polynomial time \cite{even-dar_note_2007}, the influence maximization for HC is NP-hard. We develop an approximation (greedy) algorithm for influence maximization under HC for infinite time horizon with guaranteed \emph{near-optimal} performance. Exploiting probability theory and novel Markov chain metrics, we are able to provide \emph{closed form} solution for both computing the influence spread and greedy selection step which entirely removes the need to explicitly evaluate each node as the best seed candidate; our fast and scalable algorithm, {\sc C2Greedy}, for influence maximization under HC removes the computational barrier that prevented the literature from considering the non-progressive influence models.

Our extensive experiments on several and large real and synthetic networks validate the efficiency and effectiveness of our method which outperforms the state-of-the-art in terms of both influence spread and scalability; we show that the most influential nodes under progressive models not necessarily act as the most influentials under non-progressive models and a \emph{designated} non-progressive algorithm is necessary. Moreover, we present the first real non-progressive cascade dataset which models the non-progressive propagation of research topics among network of researchers. We are planning to make this data publicly available. Our contribution in this paper is summarized as follows:
\\$\bullet$ We propose HC influence model that has favorable real world interpretations, and unifies, generalizes, and extends the existing non-progressive models and .
\\$\bullet$ We show HC has three noble key properties which enables us solving influence maximization efficiently. 
\\$\bullet$ To the best of our knowledge, we are the first to present a scalable solution for influence maximization under non-progressive LT model.
\\$\bullet$ We demonstrate high performance and scalability of our algorithm via extensive experiments and present the first ever real non-progressive cascade dataset. 

The rest of this paper is organized as follows. After a brief review on the related work, we introduce our HC model in Section \ref{ourModel}. Next, we show how to compute the influence spread for HC in closed form in Section \ref{heatEq}. In Section \ref{sec:kmip}, we present our efficient algorithm {\sc C2Greedy} for influence maximization under the HC model. Section \ref{sec:discussion} explains how HC unifies other non-progressive models and provides a more complete view of the HC model. Finally we conduct comprehensive experiments in Section \ref{sec:expr} to illustrate performance of our algorithm. 

{\bf Related work.}
After the debut of influence maximization as a data mining problem \cite{domingos_mining_2001}, it is formulated as a discrete optimization problem based on progressive influence models (LT and IC) from social and physical sciences \cite{kempe_maximizing_2003}. Kempe et al. \cite{kempe_maximizing_2003} show that influence maximization is NP-hard under LT and IC models but the influence spread is submodular for the models which enables them to use the greedy method. Although the algorithm is greedy it usually does not scale, because it needs to compute influence spread many times in each iteration while influence spread has no known closed form and is estimated by Monte Carlo simulation. The follow-up studies \cite{leskovec_cost-effective_2007,chen_scalable_2010,chen_scalable_20101,goyal2011simpath,rodriguez_uncovering_2011,due_scalable_2013} attempt to speed up this process by avoiding or decreasing the need for the MC simulation (for further details of the studies on progressive influence model please refer to Supplementary). Kempe et al. \cite{kempe_maximizing_2003} also introduce a non-progressive version of the LT influence model (NLT) and try to tackle the influence maximization problem under NLT by reducing the model to (progressive) LT, discussed in Section \ref{sec:intro}. 
\\Voter model, as the most well-known non-progressive model, is originally introduced in \cite{clifford_model_1973, holley_ergodic_1975} and adopted for viral marketing in \cite{even-dar_note_2007}. Even-Dar and Shapira show that under Voter model, highest degree nodes are the solution of influence maximization  \cite{even-dar_note_2007}. Unfortunately since the Voter model reaches consensus, i.e. one product remains in long term, it can not explain the coexistence of multiple product adoptions, which is a typical case in many real product adoptions.

\vspace*{-2pt}
\section{HEAT CONDUCTION INFLUENCE MODEL}\label{ourModel}
\vspace*{-0.05in}
The heat conduction (HC) influence model is inspired by the resemblance of influence diffusion through a social network to heat conduction through an object, where heat is transferred from the part with higher temperature to the part with lower temperature.  
We provide a simple description of HC in this section and defer the complete view of it as well as its unification property to Section \ref{sec:discussion}. 

Considering directed graph $G=(\mathcal{V},\mathcal{E})$ which represents the social (influence) network, the directed edge from node $i$ to node $j$ declares that $i$ follows $j$ (or equivalently $j$ influences $i$). Edge weight $\omega_{ij}$ indicates the amount that $i$ trusts $j$ and unless specified $0 \leq \omega_{ij} \leq 1$.  
The set of $i$'s neighbors, representing the nodes that influence $i$, is denoted by $\mathcal{N}(i)$. The influence cascade can be assumed as a \emph{binary} process in which a node who adopts the ``desired'' product is called \emph{active}, and \emph{inactive} otherwise. Note that this assumption holds for the cases with multiple products as well, where the objective is to maximize the influence (publicity) of the ``desired'' product, and the rest are all considered ``undesired''. \emph{Seed} is a node that has been selected for the direct marketing and remains active during the entire process. In HC model, the influence cascade is initiated from a set of seeds $S$ and arbitrary values for other nodes. The \emph{choice} of node $i$ to become active or inactive at time $t+1$ is a linear function of the choices of its neighbors at time $t$ as well as its intrinsic (or non-social) bias toward activeness:
\begin{equation}\label{eq:binHC}
\small Pr\big(\delta_{i}(t+1) = 1 | \mathcal{N}(i)\big) = \beta_i b + (1-\beta_i) \sum_{j \in \mathcal{N}(i)}\omega_{ij} \delta_{j}(t),
\end{equation} 
where $\beta_i \in (0, 1)$, $b \in [0,1]$, and $\sum_{j \in \mathcal{N}(i)}\omega_{ij} = 1$. Indicator function $\delta_i(t)$ is 1 when node $i$ adopts the desired product at time $t$ and 0 otherwise. We refer to (\ref{eq:binHC}) as the \emph{choice rule}. The dependence on neighbors in (\ref{eq:binHC}) represents the ``social" influence and the bias value $b$ accounts for ``non-social" influence which comes from any source out of the neighbors, e.g. media. The ``non-social" influence can explain the cases where the ``social" influence alone fails to model the cascades \cite{cha_measurement-driven_2009}.
We discuss further interpretation and extensions of HC in Section \ref{sec:discussion}.

Replacing the choice rule (\ref{eq:binHC}) in $Pr\big(\delta_{i}(t+1)\big) = \sum Pr(\delta_{i}(t+1)| \mathcal{N}(i))Pr(\mathcal{N}(i))$ results in the following  \emph{probabilistic} interpretation of the original binary HC model. Each node $i$ has a value at time $t$ denoted by $u(i,t)$ which represents the \emph{probability} that she adopts the desired product at time $t$:
\begin{equation}
u(i, t+1) = \beta_ib + (1 - \beta_i) \sum_{j \in \mathcal{N}(i)} \omega_{ij}u(j, t),
\end{equation}
Simple calculation shows that the bias value $b$ can be integrated into the network by adding a bias node $n$ (assuming that the network has $n-1$ nodes) with adoption probability $b$. Therefore, HC dynamics converts to the following: 
\begin{equation}
\label{eq:rw3}
u(i, t+1) = \sum_{j \in \mathcal{EN}(i)} P_{ij}u(j, t), 
\end{equation}
where $\mathcal{EN}(i) = \mathcal{N}(i) \cup \{n\}$ is the extended neighborhood, $P_{in} = \beta_i$, $u(n, t) = b$, and $\forall j \neq n : P_{ij} = (1-\beta_i)\omega_{ij}$. Rewriting (\ref{eq:rw3}) in the following form shows that HC follows the discrete form of \textbf{Heat Equation} \cite{lawler_random_2010}, which reveals the naming reason of HC influence model: $u(:,t+1)-u(:,t)=(P-I)u(:,t)$,
where $\mathcal{L} = I - P$ is the Laplacian matrix, $u(i,t)$ is the temperature of particle $i$ at time $t$, and ``:'' denotes the vector of all entries. 

\vspace*{-2pt}
\section{HC INFLUENCE SPREAD}\label{heatEq}
\vspace*{-0.05in}
Influence spread of set $\mathcal{S}$ for time $t$ is defined as the expected number of active nodes at time $t$ of a cascade started with $\mathcal{S}$.
Knowing that $u(i,t)$ is the probability of node $i$ being active at time $t$, \emph{influence spread} (or function) $\sigma(\mathcal{S}, t)$ is computed from:
\begin{equation}
\label{eq:inffun}
\sigma(\mathcal{S}, t)=\sum_{i\in \mathcal{V}} u(i,t).
\end{equation}
Motivated by the classical heat transfer methods, the initial and the boundary conditions should be specified to solve the heat equation and find $u(i,t)$ uniquely. 
%
In HC, the seeds $\mathcal{S}$ and the bias node are the boundary nodes and the rest are interiors.
Assuming $\mathcal{S}=\{n-1,n-2,...,n-|\mathcal{S}|\}$ and $n$ as the bias node, HC is defined by the following heat equation system:
\begin{eqnarray} \label{eq:completeEq}
\text{Main equation}&:& u(:,t+1)-u(:,t)=-\mathcal{L}u(:,t) \nonumber
\\ \text{Boundary conditions}&:& u(n,t)=b, \nonumber  
\\& & u(s,t)=1 \quad \forall s \in \mathcal{S}  
\\\text{Initial condition}&:& u(:,0)=z+ [0,...,0,\underbrace{1,...,1}_{|\mathcal{S}|},b]', \nonumber
\end{eqnarray}
where, as indicated in this formula, initial value $u(:,0)$ is the sum of two vectors: the initial values of the interior nodes ($z$) and the initial values of boundaries (the second vector). The corresponding entries of boundaries in $z$ are zero.
In the continue, exploiting probability theory and novel Markov chain metrics, we provide a closed form solution to this heat equation system.

Social network $G$ can be interpreted as an absorbing Markov chain where the absorbing states (boundary set $\mathcal{B}$) are the seeds and bias node, $\mathcal{B} = \mathcal{S} \cup \{n\}$, and $P_{ij}$ is the probability of transition from $i$ to $j$.
The adoption probability of the nodes at time $t$, i.e. $u(:,t)$, can be written as a linear function of initial condition (\ref{eq:rw3}):
\begin{equation} 
\label{eq:mainProcess}
u(:,t)=P^tu(:,0),
\end{equation}
where $P$ is row-stochastic and has the following block form: 
$P=\left[
  \begin{array}{ c c }
     R & B \\
     \mathbf{0} & I
  \end{array} \right]$. The superscript indicates the time here. 
The boundary set by definition have fixed values over time and do not follow any other nodes which leads to the zero and identity blocks $I_{(|\mathcal{S}|+1) \times (|\mathcal{S}|+1)}$. 
Blocks $R$ and $B$ represent transition probabilities of interior-to-interior and interior-to-boundary respectively. Note that different boundary conditions in (\ref{eq:completeEq}), like different seed set, result in a different $P$. Therefore both $P$ and $u(:,t)$ implicitly depend on $\mathcal{S}$.

When $t$ goes to infinity, transient part of $u$ vanishes and it converges to the steady-state solution $v = u(:,\infty)$, which is independent of time and is Harmonic, meaning that it satisfies $Pv=v$ \cite{doyle_random_1984}. Assume $v=\big( v_{\mathcal{I}}, v_{\mathcal{B}} )^T$ 
where $\mathcal{I} = \mathcal{V} \setminus \mathcal{B}$ is the set of interior nodes, then the value of interior nodes is computed from boundary nodes \cite{doyle_random_1984}:
\begin{equation} \label{eq:vQv}
v_{\mathcal{I}}=(I-R)^{-1}Bv_{\mathcal{B}} = FBv_{\mathcal{B}} = Qv_{\mathcal{B}}.
\end{equation}
where $F = (I - R)^{-1}$ is the \emph{fundamental matrix} and $F_{ij}$ indicates the average number of times that a random walk started from $i$ passes $j$ before absorption by any absorbing (boundary) nodes \cite{doyle_random_1984}. Also, the \emph{absorption probability} matrix $Q = FB$ is a $(n-|\mathcal{S}|-1)\times (|\mathcal{S}|+1)$ row-stochastic matrix, where $Q_{ij}$ denotes the probability of absorption of a random walk started from $i$ by the absorbing node $j$ \cite{doyle_random_1984}.

From here on, without loss of generality, we assume $b$ to be zero in equation (\ref{eq:completeEq}). Using (\ref{eq:mainProcess}) and (\ref{eq:vQv}), the influence spreads for infinite time can be computed in closed form:
\begin{equation}\label{eq:generalInf}
\sigma(\mathcal{S},\infty) = \sum_{i=1}^n v(i) = |\mathcal{S}| + \sum_{i\in \mathcal{I}}\sum_{s\in \mathcal{S}} Q^{\mathcal{S}}_{is}.
\end{equation}
The superscript in $Q^{\mathcal{S}}$ and $P^{\mathcal{S}}$ explicitly indicates that they are functions of seed set $\mathcal{S}$. Note that in fact they are depending on the total boundary set, $\mathcal{B}=\mathcal{S}\cup \{n\}$, but since the bias node is always a boundary, throughout this paper we discard it from the superscripts to avoid clutter.

\vspace*{-2pt}
\section{INFLUENCE MAXIMIZATION FOR HC}\label{sec:kmip}
\vspace*{-0.05in}
In this section we solve the influence maximization problem for \textit{infinite time horizon} under HC model, formulated as follows: 
\begin{equation}
\label{eq:opt}
\mathcal{S^*}=\argmax_{\mathcal{S}\subseteq \mathcal{V}} \sigma(\mathcal{S}, \infty), \qquad s.t. \qquad  |\mathcal{S}| \leq K.
\end{equation}

\subsection{INFLUENCE MAXIMIZATION FOR $K=1$}
\vspace*{-0.1in}

Based on (\ref{eq:generalInf}) and (\ref{eq:opt}), the most influential person (MIP) is the solution of the following optimization problem: $\argmax_{\mathcal{V} \setminus \{n\}}  \sum_{i\in \mathcal{V} \setminus \{s,n\}} Q^{\{s\}}_{is}$.  
This equation states that to find the MIP, we need to pick each candidate $s$ and make it absorbing and compute the new $P$ as $P^{\{s\}}$ which in turn changes $Q$ to $Q^{\{s\}}$, and repeat this procedure $n - 1$ times for all $s$. This procedure is problematic because for each $Q^{\{s\}}$ we require to recompute matrix $F^{\{s\}}$ which involves matrix inversion. 
But, in the following theorem we show that we are able to do this by only one matrix inversion instead of $n - 1$ matrix inversions, and having matrix $F^{\emptyset}$ is enough to find the most influential person of the network ($\emptyset$ sign indicated no seed is selected):
\begin{theorem}
\label{MIP_infinite}
MIP under HC (\ref{eq:binHC}) when $t\rightarrow\infty$ can be computed in closed form from the following formula:
\begin{equation}
MIP=\argmax_{s\in \mathcal{V} \setminus \{n\}}  \sum_{i \in \mathcal{V} \setminus \{n\}} \frac{F^{\emptyset}_{is}}{F^{\emptyset}_{ss}}=\argmax  \mathbf{1'} \breve{F}^{\emptyset},
\end{equation}
where $\breve{F}^{\emptyset}$ is $F^{\emptyset}$ when each of its columns is normalized by the corresponding diagonal entry. Note that left multiplication of all ones row vector is just a column-sum operation.
\end{theorem}

\subsection{INFLUENCE MAXIMIZATION FOR $K>1$}
\vspace*{-0.1in}
Although the influence maximization can be solved optimally for $K=1$ , the general problem (\ref{eq:opt}) under HC for $K>1$ is NP-hard:
\begin{theorem}
Given a network $G = (\mathcal{V},\mathcal{E})$ and a seed set $\mathcal{S}\subseteq \mathcal{V}$, influence maximization for \textit{infinite time horizon} (\ref{eq:opt}) under HC defined by (\ref{eq:binHC}) is NP-hard.
\end{theorem}

In spite of being NP-hard, we show that the influence spread  $\sigma(\mathcal{S},\infty)$ is \emph{submodular} in the seed set $\mathcal{S}$ which enables us to find a provable near-optimal greedy solution. A set function $f : 2^\mathcal{V} \rightarrow \Real$ maps subsets of a finite set $\mathcal{V}$ to the real numbers and is submodular
if for $\mathcal{T} \subseteq \mathcal{S} \subseteq \mathcal{V}$ and $s \in \mathcal{V} \setminus \mathcal{S}$, $f(\mathcal{T}\cup\{s\})- f(\mathcal{T}) \geq f(\mathcal{S} \cup \{s\})- f(\mathcal{S})$ holds, which is the diminishing return property. Following theorem presents our established submodularity results.
\begin{theorem}
Given a network $G = (\mathcal{V},\mathcal{E})$, influence spread $\sigma(\mathcal{S},\infty)$ under HC model is non-negative monotone submodular function.
\end{theorem}
The greedy solution adds nodes to the seed set $\mathcal{S}$ sequentially and maximizes a monotone submodular function with $(1 - 1/e)$ factor approximation guarantee \cite{nemhauser_analysis_1978}. More formally the $(k+1)$-th seed is the node with maximum \textbf{marginal gain}: $(k+1)\text{th-} MIP_t = \argmax_{s \in \mathcal{V} \setminus \{\mathcal{S}_{k} \cup \{n\}\}} \sigma(\mathcal{S}_{k}\cup \{s\},t)-\sigma(\mathcal{S}_{k},t)$, where $\mathcal{S}_{k}$ is the set of $k$ seeds which have been picked already. 
Although we can compute the above objective function in closed form, for selecting the next seed we have to test all $s$ to solve the problem which is the approach of all existing greedy based method in the literature. Previously a lazy greedy scheme have been introduced to reduce the number testing candidate nodes $s$ \cite{leskovec_cost-effective_2007}. In the next section we go one step further and show that under HC model and for \textit{infinite time horizon} we can solve the marginal gain in \emph{closed form}. 

\subsection{GREEDY SELECTION}\label{sub:greedySel}
\vspace*{-0.1in}
An important characteristic of the linear systems, like HC when $t \rightarrow \infty$, is the ``superposition" principle. We leverage this principle to calculate the marginal gain of the nodes efficiently and pick the one with maximum gain for the greedy algorithm. 
Based on this principle, the value of each node in HC for infinite time, and for a given seed set $\mathcal{S}$, is equal to the algebraic sum of the values caused by each seed acting alone, while all other values of seeds have been kept zero. 
Therefore, when a node $s$ is added to the seed set $\mathcal{S}_k$, its marginal gain can be calculated as the summation of values of the nodes when all of the values of $\mathcal{S}_k$ have been turned to zero and node $s$ is the only seed in the network, whose value is $1-v^{\mathcal{S}_k}(s)$. 
In this new problem, the vector of boundary values $v_{\mathcal{B}}^{\mathcal{S}_k \cup \{s\}}$ is a vector of all 0's except the entry corresponding to the node $s$ with value $1-v^{\mathcal{S}_k}(s)$, and the value of interior node $i$ is obtained from (\ref{eq:vQv}):
\begin{equation} \label{eq:vgreedy}
v_\mathcal{I}^{\mathcal{S}_k \cup \{s\}}(i)=Q_{is}^{\mathcal{S}_k \cup \{s\}}(1-v^{\mathcal{S}_k}(s)) \nonumber
\end{equation}
Substituting $Q$ from lemma 3 result (see Supplementary), the $k+1$-th seed is determined from the following closed form equation:
\begin{eqnarray} \label{eq:k+1}
&&(k+1)\text{th-} MIP \nonumber
\\&=&\argmax_{s \in \mathcal{V} \setminus \{\mathcal{S}_{k}\cup\{n\}\}} \sum_{i \in \mathcal{V} \setminus \{\mathcal{S}_{k}\cup\{n\}\}} \frac{{F_{is}^{\mathcal{S}_{k}}}}{F_{ss}^{\mathcal{S}_{k}}}\big(1-v^{\mathcal{S}_{k}}(s)\big), \nonumber
\\&=& \argmax (\mathbf{1}-v^{\mathcal{S}_{k}})'\breve{F}^{\mathcal{S}_{k}} 
\end{eqnarray}
Note that vector $v^{\mathcal{S}_{k}}$ is obtained in step $k$ and is known, and matrix $F^{\mathcal{S}_k}$ can be calculated from $F^{\mathcal{S}_{k-1}}$ without any need for matrix inversion (see Supplementary, lemma 1). 
One may observe that equation (\ref{eq:k+1}) is the general form of Theorem 1, since $v^{\mathcal{S}_0}=v^{\emptyset}= {0}$. Notice that equation (\ref{eq:k+1}) intuitively uses two criteria for selecting the new seed: its current value should be far from 1 (higher value for $(1-v^{\mathcal{S}_k}(s))$ term) which suggests that it is far from the previously selected seeds, and at the same time it should  have a high network centrality (corresponding to the ${F_{is}^{\mathcal{S}_{k}}}/F_{ss}^{\mathcal{S}_{k}}$ term). Algorithm \ref{alg1} summarizes our {\sc C2Greedy} method for $t \to \infty$: a greedy algorithm with 2 closed form steps. Operator $\otimes$ in step 10 denotes the Hadamard product.

\begin{algorithm}                      
	\caption{ \sc C2Greedy }
	\label{alg1}                           
	\begin{algorithmic}[1]
	\small                    
		\STATE {\bfseries input:} extended directed network $G=(\mathcal{V},\mathcal{E})$ with bias node $n$, maximum budget $K$.
		\STATE {\bfseries output:} seed set $\mathcal{S}_K \subseteq \mathcal{V}$ with cardinality $K$.
		\STATE compute matrix $P$ from $G$.
		\STATE $\mathcal{S}_0 := \emptyset$
		\STATE $F^{\mathcal{S}_0} := (I - P^{\mathcal{S}_0})^{-1}$
		\STATE $s = \argmax  \mathbf{1'} \breve{F}^{\emptyset}$, and $\mathcal{S}_1 = \mathcal{S}_0 \cup \{s\}$
		\STATE $v^{\mathcal{S}_{1}}=\breve{F}^{\mathcal{S}_{0}}(:,s)$
		\FOR {$k = 1$ to $K-1$}			
			\STATE $\forall i, j \in \mathcal{I}: F_{ij}^{\mathcal{S}_k \cup \{s\}}=F^{\mathcal{S}_k}_{ij}-\frac{F^{\mathcal{S}_k}_{is}F^{\mathcal{S}_k}_{sj}}{F^{\mathcal{S}_k}_{ss}}$
			\STATE $s = \argmax (\mathbf{1}-v^{\mathcal{S}_{k}})' \otimes\mathbf{1'}\breve{F}^{\mathcal{S}_{k}}$, and $\mathcal{S}_{k+1} = \mathcal{S}_k \cup \{s\}$
            \STATE $v^{\mathcal{S}_{k+1}}=v^{\mathcal{S}_{k}}+(\mathbf{1}-v^{\mathcal{S}_{k}}(s))\breve{F}^{\mathcal{S}_{k}}(:,s)$
		\ENDFOR
	\end{algorithmic}
\end{algorithm}
\vspace*{-0.1in}

\begin{table*}[t]
\centering	
\caption{\small Specifying the equal heat system for existing non-progressive influence models.}
\label{tb:discussion}
\vskip 0in
\begin{center}
\begin{small}
\begin{tabular}{| c | c | c || c | c | c | c || c |}
\cline{1-8}
\multicolumn{1}{ |c| }{\multirow{2}{*}{Model}} & \multicolumn{1}{ |c| }{\multirow{2}{*}{\specialcell{Non-Social \\ influence}}}  & \multicolumn{1}{ |c|| }{\multirow{2}{*}{Weighted edges}} & \multicolumn{2}{|c|}{Boundary} & \multicolumn{2}{|c||}{Init. Cond.} & \multicolumn{1}{|c|}{\multirow{2}{*}{\specialcell{Equivalent Physical \\ Heat Conduction System}}}  \\ 
\cline{4-7}
& & & High $T = 1$ & Low $T < 1$ & $=0$ & $\neq 0$ &  \\
\cline{1-8}
NLT1 & $\surd$ & $\surd$ & & $\surd$ & $\surd$ &  & \specialcell{Circular ring with a fixed-temp. point}\\
\cline{1-8}
NLT2 & $\surd$ & $\surd$ & $\surd$ & $\surd$ & $\surd$ & & \specialcell{A rod with fixed-temp.\\ends, one high one low}\\
\cline{1-8}
NLT3 & & $\surd$ & & & $\surd$ &  & \specialcell{(Isolated) circular ring}\\
\cline{1-8}
NLT4 & & $\surd$ & $\surd$ & & $\surd$ & &  \specialcell{Circular ring with a fixed-temp. point}\\
\cline{1-8}
Voter & & & & & & $\surd$ & \specialcell{(Isolated) circular ring} \\
\cline{1-8}
GLT & $\surd$ & $\surd$ & & $\surd$ & & &  \specialcell{Circular ring with a fixed-temp. point}\\
\cline{1-8}
\end{tabular}
\end{small}
\end{center}
\vskip -0.25in
\end{table*}

\vspace*{-2pt}
\section{DISCUSSION}\label{sec:discussion}
\vspace*{-0.05in}

In this section, we present the comprehensive view of HC model and elaborate its (unifying) relation to the other models by providing multiple interpretations.

{\bf Social interpretation. }
HC can be simply extended to model many real cases that the other influence models fail to cover.
As briefly mentioned in Section \ref{ourModel}, the original HC (\ref{eq:binHC}), models both ``social" and ``non-social" influences which cover the observations from the real datasets \cite{cha_measurement-driven_2009}. The extension of HC which is more flexible in modeling real world cascades is as follows:
\begin{equation}
\label{eq:rw1}
\small
u(i, t+1) = m\alpha_i + r\gamma_i + (1 - \gamma_i-\alpha_i) \sum_{j \in \mathcal{N}(i)} \omega_{ij}u(j, t),
\end{equation}
where, $\sum_{j \in \mathcal{N}(i)} \omega_{ij} = 1$, $\gamma_i,\alpha_i \in [0,1]$, $m=1$, and $r=0$. 
Factor $r$ models the ``discouraging" factor like intrinsic \emph{reluctance} of customers toward a new product, and $m$ represents ``encouraging" factor like \emph{media} that promotes the new product. These two factors can explain cases where all neighbors of a node are active but the node remains inactive, or when a node becomes active while none of her neighbors are active \cite{cha_measurement-driven_2009}. Note that all of the formulas and results stated so far is simply applicable to the general HC model (\ref{eq:rw1}).

{\bf Unification of existing non-progressive models. }
HC (\ref{eq:binHC}) unifies and extends many of the existing non-progressive models. 
In the Voter model, a node updates its choice at each time step by picking one of its neighbors randomly and adopting its choice. 
In other words, the choice rule of node $i$ is the ratio of the number of her active neighbors to her total number of neighbors. 
Thus, Voter's choice rule is the simplified form of HC's choice rule (\ref{eq:binHC}) where  $\omega_{ij}$ is equal to $\frac{1}{d_i}$ ($d_i$ is the out-degree of node $i$) and all $\beta_i$s are set to zero. Also, note that having $\beta_i=0$ indicates that the Voter does not cover the ``non-social" influence. 
 
In the \emph{non-progressive} LT (NLT) \cite{kempe_maximizing_2003}, each node is assigned a random threshold $\theta$ \emph{at each time step} and becomes active if the weighted number of its active neighbors (at previous time step) becomes larger than its threshold: $\sum_{j \in \mathcal{N}(i)} \omega_{ij}\delta_j(t) \geq \theta_i(t+1)$, where the edge weights satisfy $\sum_{j \in \mathcal{N}(i)}\omega_{ij} \leq 1$. 
Thus, the choice rule of node $i$ at time $(t+1)$ under the NLT is obtained from the following equation:
\begin{eqnarray} \label{eq:NLT}
Pr\big(\delta_{i}(t+1) = 1 | \mathcal{N}(i)\big) &=&Pr\big(\theta_i(t+1) \leq \Sigma \omega_{ij}^{\text{\tiny{NLT}}} \delta_{j}(t)\big) \nonumber
\\&=& \Sigma \omega_{ij}^{\text{\tiny{NLT}}} \delta_{j}(t), 
\end{eqnarray}
where the second equality is the result of sampling $\theta_i(t+1)$ from the \textit{uniform distribution} $U(0,1)$. Equation (\ref{eq:NLT}) is the simplified form of HC's choice rule   (\ref{eq:binHC}), where $b=0$ and $(1-\beta_i)\omega_{ij}^{\text{\tiny{HC}}}=\omega_{ij}^{\text{\tiny{NLT}}}$. Note that since in the NLT $b$ accepts only zero value, this influence model also cannot cover \textit{encouraging} ``non-social" influence. Moreover, if the edge weights' gap in NLT, i.e. $g_i=1-\sum_{j \in \mathcal{N}(i)}\omega_{ij}^{\text{\tiny{NLT}}}$, is zero for all the nodes, it cannot model the ``non-social" influence at all, since the corresponding $\beta_i$'s in (\ref{eq:binHC}) would be equal to zero in that case.

Generalized linear threshold (GLT) is another non-progressive model proposed in \cite{pathak_generalized_2010} to model the adoption process of \textit{multiple} products. Assigning a color $c \in \mathcal{C}$ to each product, a node updates its color, at each time step, by randomly picking one of its neighbors based on its edge weights and adopts the selected neighbor's color. For binary case $|\mathcal{C}|=2$, where we only distinct between adoption of a desired product (active) and the rest of products (inactive), GLT's choice rule reduces to the following equation: $Pr\big(\delta_{i}(t+1) = 1 | \mathcal{N}(i)\big) = \frac{\beta}{2} + (1-\beta) \sum_{j \in \mathcal{N}(i)}\omega_{ij} \delta_{j}(t)$.
It is easy to see that this is the restricted form of HC's choice rule  (\ref{eq:binHC}), where nodes are all connected to the bias node with equal weight of $\beta$ and bias value $b$  has to be $\frac{\beta}{2}$.

{\bf Physical interpretation. }
We showed that the existing non-progressive models are special cases of HC, and in this part we describe their equal heat conduction system which are uniquely specified by the initial and boundary conditions. 
Table \ref{tb:discussion} summarizes the heat interpretation of the influence models. 
We introduce four variants of non-progressive LT, based on two factors: seed and gap $g_i$. NLT1 and NLT2 support non-zero gaps, and NLT2 and NLT4 allows seeds, i.e. nodes in the network that always remain active. 
The non-progressive LT model presented in \cite{kempe _maximizing_2003} is equivalent to NLT2. 
Reluctance factor and seeds in all models are equivalent to the low and high temperature boundaries respectively, and initial condition addresses the interiors' initial values ($z$ in (\ref{eq:completeEq})). The non-social influence and edge weights factors appear in the Laplacian matrix calculation of (\ref{eq:completeEq}). 
The equivalent physical heat conduction systems are easy to understand, here we just briefly point out the equivalence of the Voter model and the isolated circular ring. Circular ring is a rod whose ends are connected to each other and do not have any energy exchange with outside \cite{incropera2011fundamentals} which explains why the Voter conserves the total initial heat energy, and reaches to an equilibrium with an equal temperature for all of the nodes, i.e., consensus. 

{\bf Random walk interpretation. }
Beside the heat conduction view, the random walk prospect helps to gain a better understanding of the models and their relations. 
Assume that active and inactive nodes are colored black and white respectively. 
Consider the \emph{original view} of any influence model which is the actual process that unfolds in time, so we look at the time-forward direction. 
We take a snapshot of the colored network at each time step $t$. 
Putting together the sequence of snapshots, the result is a random walk in the ``colored graphs'' state space with $2^n$ states. 
On the other hand, the \emph{dual view} looks at the time-reverse direction of influence models. 
It is known for both IC-based models (like IC \cite{kempe _maximizing_2003} and ConTinEst \cite{due_scalable_2013}) and LT-based models (Table \ref{tb:discussion} as well as HC and LT)  that a single node from $\mathcal{N}(i)$ is responsible for $i$'s color switch, which we name it as the parent of $i$. 
Now assuming that the process has advanced up to the time $t$, we reverse the process by starting from each node $i$ and follow its ancestors. 
Here is the point where IC and LT based models separate from each other: due to $\sum_{j \in \mathcal{N}(i)} \omega_{ij} \leq 1$ constraint, ancestors of $i$ in the LT-based models form a random walk starting from node $i$, which is not the case in IC-based models. 
Note that we have $n$ random walks that can meet and merge, thus they are known as \emph{coalescing random walks} \cite{aldous_reversible_2002}. 
This view also helps us to demonstrate the essential difference between progressive and non-progressive models. 
Dual view of progressive LT model is a \emph{coalescing self-avoiding walks} which is the outcome of randomizing the threshold $\theta$ only once at the beginning of the process for the nodes in each realization. 
This bounds the number of ``live" edges \cite{kempe_maximizing_2003} connected to each node by one which prevents the creation of ``loop" in the influence paths. Note that both counting and finding the probability of self-avoiding walks are $\#P$ hard \cite{chen_scalable_2010}.

\vspace*{-5pt}
\section{EXPERIMENTS} \label{sec:expr}
\vspace*{-0.05in}

In this section, we examine several aspects of \textsc{C2Greedy} and compare it with state-of-the-art methods. Experiments mainly focus on influence maximization and timing aspects. Finally, we present one example of real non-progressive data and illustrate the result of \textsc{C2Greedy}.

\subsection{DATASET}
\vspace*{-0.1in}
Table \ref{tb:netList} summarizes the statistics of the networks that we use throughout the experiments. We work with both synthetic and real networks which we briefly discuss next.

{\bf Synthetic network generation.}
We consider the following types of Kronecker network for extensive performance comparison of our method with the state-of-the-art methods: random \cite{erdos_evolution_1960} (parameter matrix $[0.5, 0.5; 0.5, 0.5]$), hierarchical \cite{clauset_hierarchical_2008} ($[0.9, 0.1; 0.1; 0.9]$), and core-periphery \cite{leskovec_kronecker_2010} ($[0.9, 0.5; 0.5, 0.3]$). 
We generate 10 samples from each network and report the average performance of each method. 
Edge weights are drawn uniformly at random from $[0,1]$ and weights of each node's outgoing edges is normalized to 1. 
For timing experiment, we use ForestFire \cite{clauset_hierarchical_2008} (Scale-free) network with forward and backward burning probability of 0.35 and 0.25, respectively, and set the outgoing edge weights of node $i$ to $1/|\mathcal{N}(i)|$.
The expected density, i.e., number of edges per node, for the resulted ForestFire networks is 2.5.

\begin{table}[t]
\vskip -0.3in
\centering	
\caption{\small List of networks used in experiments.}
\label{tb:netList}
\begin{center}
\begin{small}
\vskip -0.1in
\begin{tabular}{c c | c | c | c |}
\cline{3-5}
& & $|\mathcal{V}|$ & $|\mathcal{E}|$  & Params \\ 
\cline{1-5}
\multicolumn{1}{ |c| }{\multirow{4}{*}{\specialcell{Synthetic \\ Networks}}} & Random &  1024 & - & {\tiny$[0.5, 0.5; 0.5, 0.5]$} \\
\cline{2-5}
\multicolumn{1}{ |c| }{} & Hier. & 1024 & - &  \tiny{$[0.9, 0.1; 0.1; 0.9]$}  \\
\cline{2-5}
\multicolumn{1}{ |c| }{} & Core. & 1024 & - & \tiny{$[0.9, 0.5; 0.5, 0.3]$}   \\
\cline{2-5}
\multicolumn{1}{ |c| }{} & ForestFire & $1$K-$300$K & $2.5 |\mathcal{V}|$ & \tiny{$[0.35, 0.25]$} \\
\cline{1-5}
\multicolumn{1}{ |c| }{\multirow{4}{*}{\specialcell{Real \\ Networks}}} &  KClub & 34 & 501 & -   \\
\cline{2-5}
\multicolumn{1}{ |c| }{}& PBlogs & 1490 & 19087 & -  \\
\cline{2-5}
\multicolumn{1}{ |c| }{}& WikiVote & 7115 & 103689 & - \\ 
\cline{2-5}
\multicolumn{1}{ |c| }{}& MLWFW & 10604 & 168918 & - \\
\cline{1-5}
\end{tabular}
\end{small}
\end{center}
\vskip -0.3in
\end{table}

{\bf Real Networks.}
Zachary's karate club network (KClub) is a small friendship network  with 34 nodes and 501 edges \cite{zachary_information_1977}. The political blogs network (PBlogs) \cite{adamic_political_2005}, is a moderate size directed network of hyperlinks between weblogs on US politics with 1490 nodes and 19087 edges. Wikipedia vote network (WikiVote), is the network
of who-vote-whom from wikipedia administrator elections
\cite{leskovec_predicting_2010} with 7115 nodes and 103689 edges. Finally, MLWFW is the network of who-follow-whom in the machine learning research community which we extract from citation networks of combined ACM and DBLP citation network which is available as a part of ArnetMiner  \cite{tang_arnetminer:_2008}. For more information about MLWFW refer to Section \ref{subsec:real}. 

For all synthetic and real networks, after constructing the network, we add the bias node to the network and connect all nodes to it with weight $\beta_i=0.1$ and re-normalize the weight of the other edges accordingly.

\subsection{INFLUENCE MAXIMIZATION}
\vspace*{-0.1in}
In this section we investigate the performance of \textsc{C2Greedy} in the main task of influence maximization i.e., solving the set function optimization (\ref{eq:opt}). Since finding the optimal solution for (\ref{eq:opt}) is NP-hard, we compare \textsc{C2Greedy} with optimal solution only for a small network, then for a large network we show that \textsc{C2Greedy} result is close to the online bound \cite{leskovec_cost-effective_2007}.  We also compare the performance of \textsc{C2Greedy} with the state-of-the-art methods proposed for solving (\ref{eq:opt}) under different (mostly progressive) influence models.

\textbf{\textsc{C2Greedy} vs. optimal.} 
For testing the quality of \textsc{C2Greedy} method, we compare its performance with the best seed set (determined by brute force) on a small size network. 
We work with the KClub network for the brute-force experiment with $K = 5$. 
As Figure \ref{fig:bruteForce} shows \textsc{C2Greedy} selects nodes that match the performance of the optimal seed set. 
In the next step, on a larger network, we show that the performance of \textsc{C2Greedy} is close to the known online upper bound \cite{leskovec_cost-effective_2007}.
We compute the online and offline bounds of greedy influence maximization \cite{leskovec_cost-effective_2007} with $K=30$ for PBlogs network. 
Figure \ref{fig:onlineBound} illustrates that \textsc{C2Greedy} result is close to the online bound and therefore close to the optimal solution's performance.

\begin{figure}
\vskip -0.1in
\begin{center}
\subfigure[C2Greedy vs. optimal.]{\label{fig:bruteForce}\includegraphics[scale=.21]{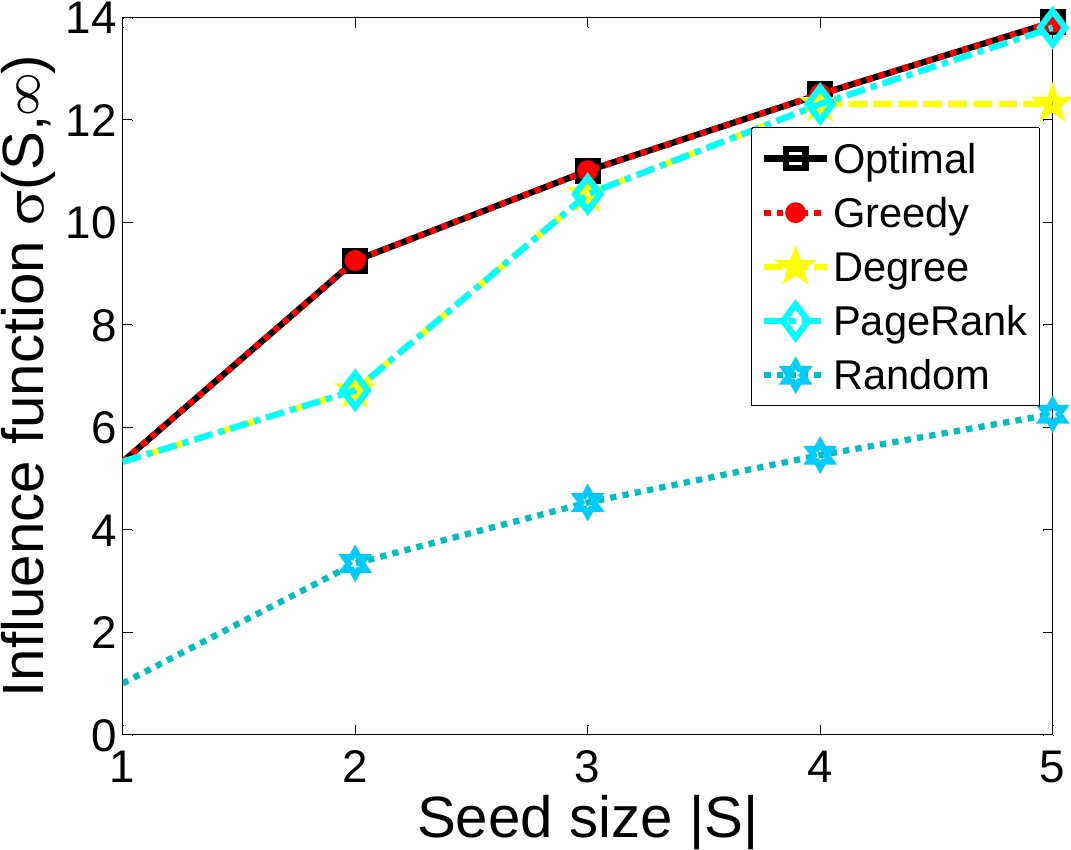}} ~
\subfigure[Online and offline bounds.]{\label{fig:onlineBound}\includegraphics[scale=.21]{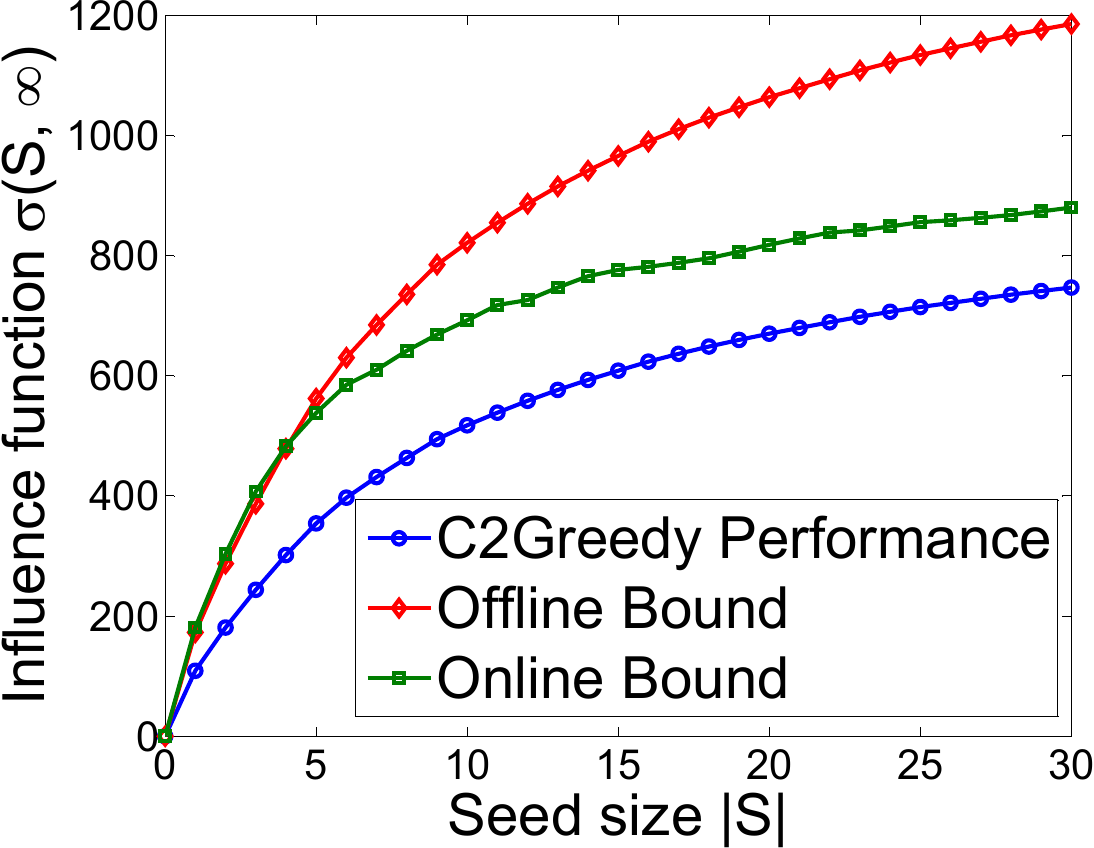}}
\vskip -0.1in
\caption{\small For small network (a) shows C2Greedy matches the optimal performance. For a larger network (b) compares performance of C2Greedy with online and offline bounds.}
\label{fig:optimal}
\end{center}
\vskip -0.2in
\end{figure}

\textbf{\textsc{C2Greedy} vs. state-of-the-art.}
Next, we compare \textsc{C2Greedy} with the state-of-the-art methods of influence maximization over three aforementioned synthetic networks and WikiVote real network. Among baseline methods PMIA \cite{chen_scalable_20101} and LDAG \cite{chen_scalable_2010} are approximation for IC and LT models respectively and SP1M \cite{kimura_tractable_2006} is a shortest-path based heuristic algorithm for influence maximization under IC. ConTinEst \cite{gomez-rodriguez_influence_2012} is a recent method for solving continuous time model of \cite{rodriguez_uncovering_2011} and PageRank is the well-known information retrieval algorithm \cite{brin_anatomy_1998}. Finally, Degree selects the nodes with highest degree as the most influential and Random picks the seed set randomly. 

The comparison results are depicted in Figure \ref{fig:compare}. Interestingly, our algorithm outperformed all of the baselines. Strangely, ConTinEst performs close to Random (except in the random network). 
A closer look at the results for three synthetic networks reveal that except ConTinEst's odd behavior all other methods have persistence rank in performance. \textsc{C2Greedy} is the best method and is followed by PMIA and LDAG, both in second place, which are closely followed by SP1M. PageRank, Degree and Random are next methods in order. In WikiVote real network of Figure \ref{fig:realCmp} surprisingly most of the state-of-the-art methods perform terribly poor and Degree (as the KMIP solution to Voter model) is the only competitor of \textsc{C2Greedy}. Result of experiment with WikiVote shows that most influential nodes in a progressive models are not necessary influential in non-progressive ones, and designing non-progressive-specific algorithms (like \textsc{C2Greedy}) is required for influence maximization under non-progressive models.

\begin{figure*}
\vskip -0.1in
\begin{center}
\subfigure[Random network]{\label{fig:randCmp}\includegraphics[scale=.21]{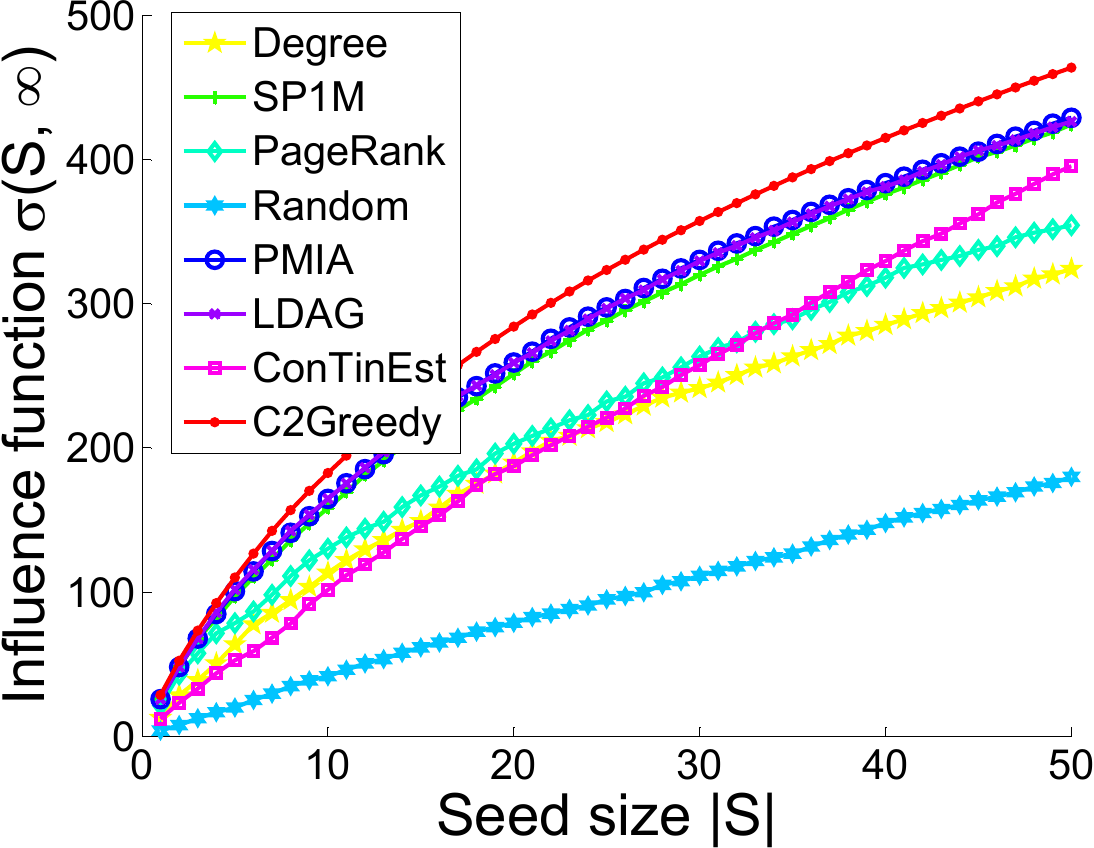}} ~
\subfigure[Hierarchical network]{\label{fig:hierCmp}\includegraphics[scale=.21]{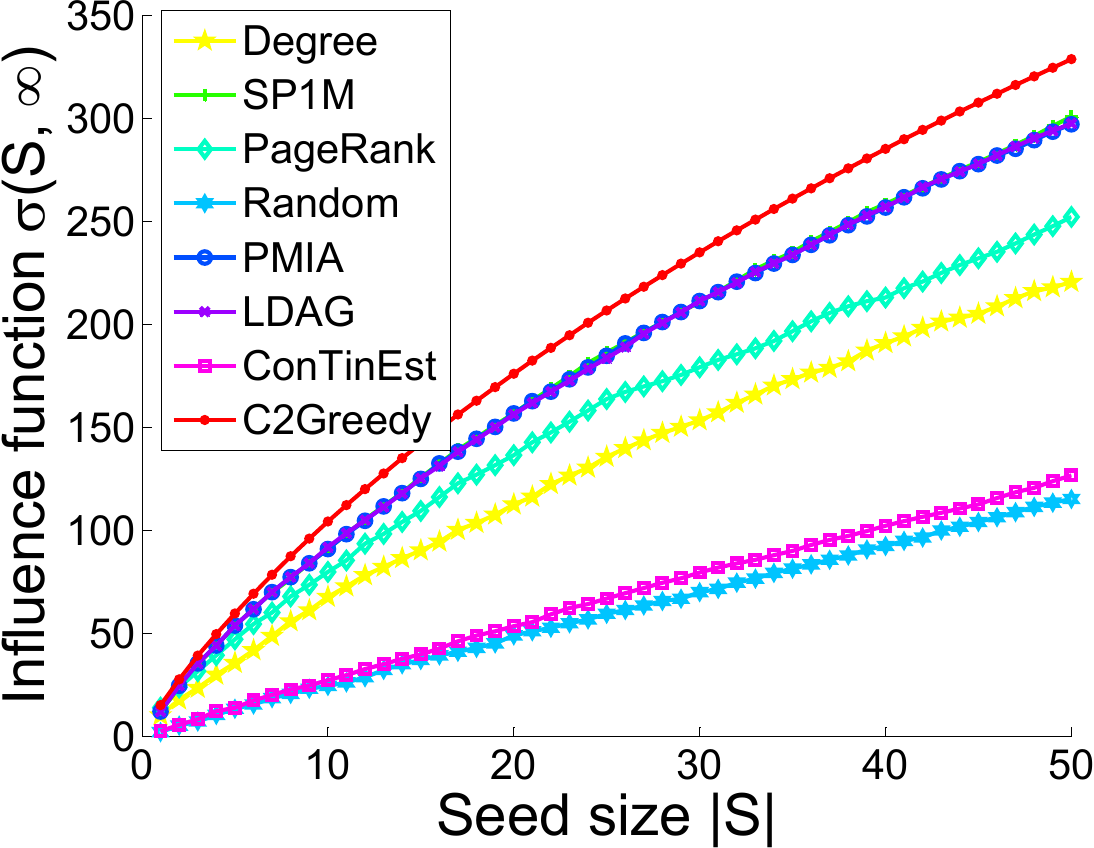}}
\subfigure[Core-periphery network]{\label{fig:corCmp}\includegraphics[scale=.21]{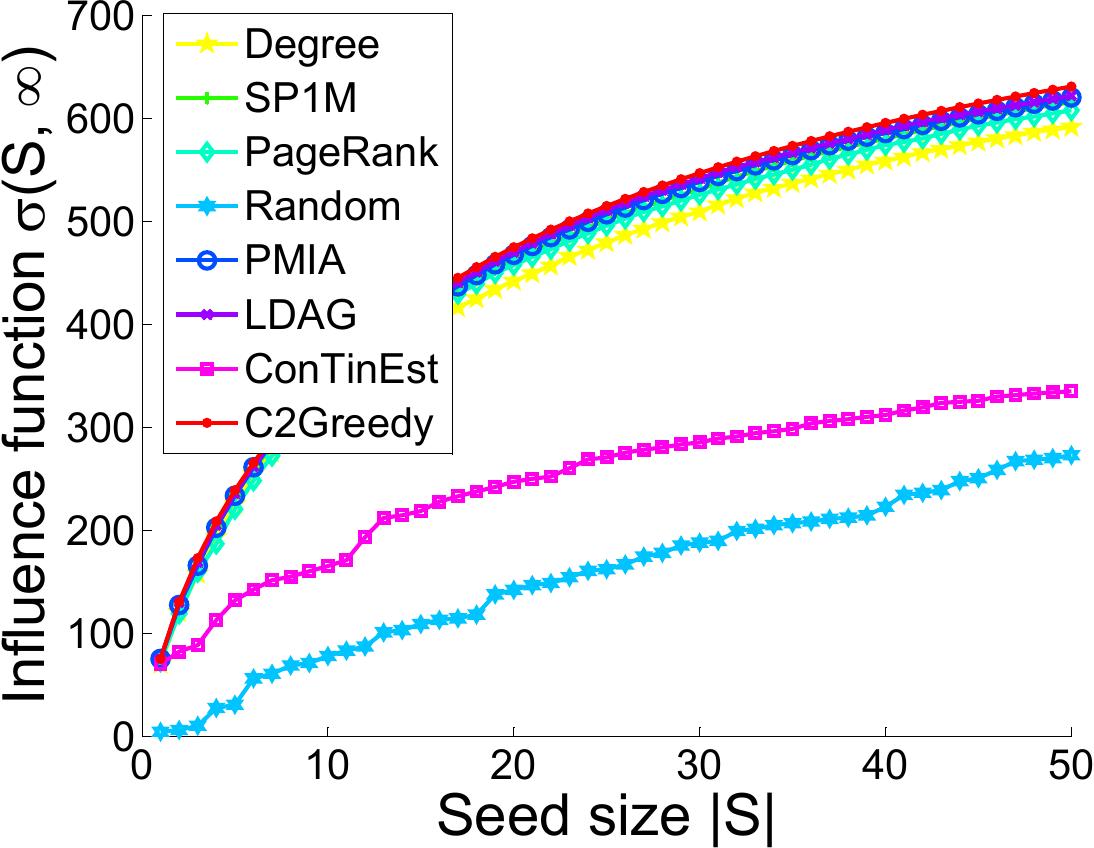}} ~
\subfigure[Real network (WikiVote)]{\label{fig:realCmp}\includegraphics[scale=.21]{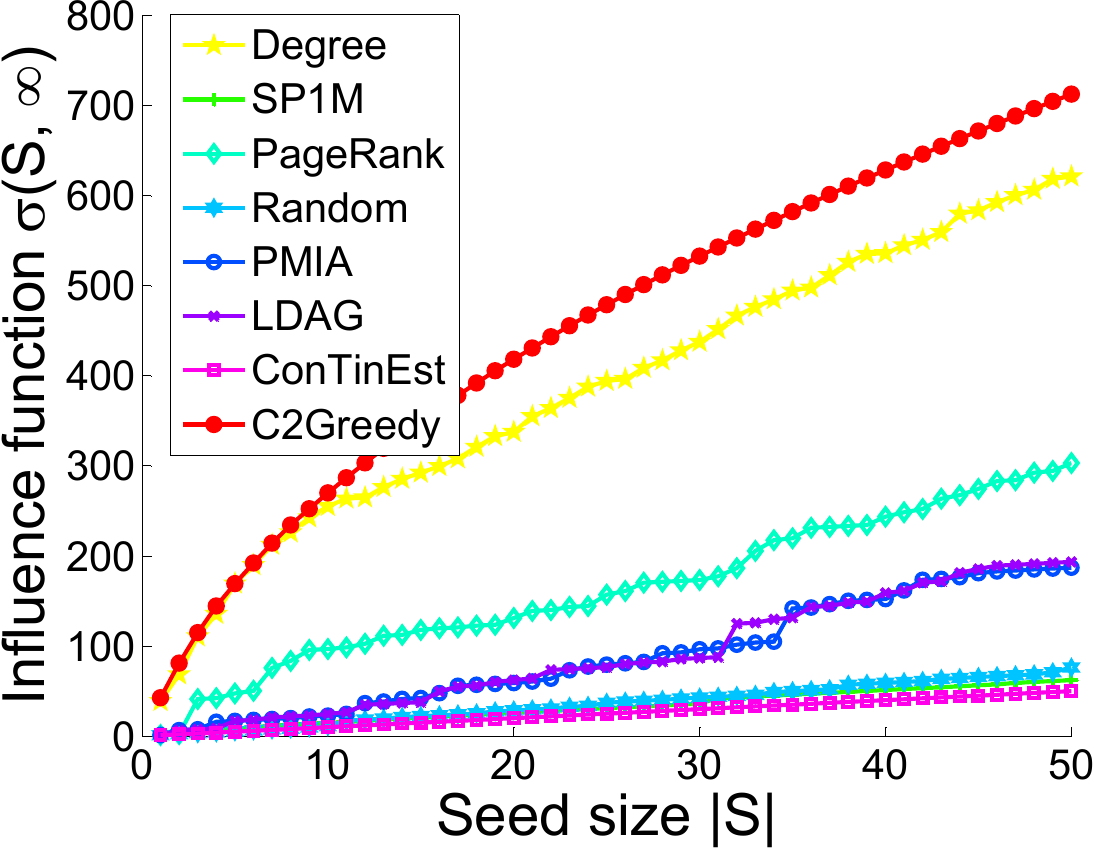}}
\vskip -0.1in
\caption{\small Comparing performance of C2Greedy with state-of-the-art influence maximization methods. Networks of (a), (b), and (c) are synthetic and (d) is a real network.}
\label{fig:compare}
\end{center}
\vskip -0.15in
\end{figure*}

\subsection{SPEED AND SCALABILITY}
\vspace*{-0.1in}
In this part we illustrate the speed benefits of having two closed form updates in the greedy algorithm and also deal with the required single inverse computation of \textsc{C2Greedy} to prove the scalability of our method. 

{\bf Closed form benefits.}
As discussed in Section \ref{sec:kmip}, our main algorithm \textsc{C2Greedy} benefits from closed form computation for both  influence spread (\ref{eq:generalInf}) and greedy selection (\ref{eq:k+1}). 
To show the gain of these closed form solutions, we run the greedy algorithm in three different settings. 
First without using any of (\ref{eq:generalInf}) and (\ref{eq:k+1}) which we call \textsc{Greedy} and uses Monte Carlo simulation to estimate the influence spread. 
Second we only use (\ref{eq:generalInf}) to have the closed form for influence spread without closed form greedy update of (\ref{eq:k+1}) which results in \begin{sc}C1Greedy\end{sc}, and finally \textsc{C2Greedy} which uses both (\ref{eq:generalInf}) and (\ref{eq:k+1}).
Note that we can add lazy update of \cite{leskovec_cost-effective_2007} (see Supplementary) to \begin{sc}Greedy\end{sc} and \begin{sc}C1Greedy\end{sc} to get \begin{sc}LGreedy\end{sc} and \begin{sc}LC1Greedy\end{sc} respectively.
Finally we include the original greedy method \cite{kempe_maximizing_2003} of solving LT model (progressive version of our model) and its lazy variant, with 100 iteration of Monte Carlo simulation. 
Note that for having a good approximation of influence spread in LT model, simulations are run for several thousand iterations, but here we just want to illustrate that the greedy algorithm for HC is much faster than LT, for which 100 iterations is enough. 
Figure \ref{fig:bars} illustrates the speed in log-scale of all seven algorithms for $K=10$ over the Pblogs dataset \cite{adamic_political_2005}. 
Note that the required time of inverse computation (\ref{eq:vQv}) is also included. 
The results confirm that both closed forms decrease the timing \textit{significantly} (1 sec vs. 461 sec for the next best variation) and help the greedy algorithm far more than the lazy update.

{\bf Per-seed selection time.}
The major computational bottleneck of our algorithm is the inverse computation of (\ref{eq:vQv}). But fortunately this is needed once and at the beginning of the process. Here assuming offline inverse computation, we are interested in the cost of adding each seed. Figure \ref{fig:amortized} compares the cost of selecting $k$-th seed for the five variation of our algorithm, plus LT and LazyLT all described  previously. As expected \textsc{C2Greedy} requires the lowest computation time per seed. 
Also, the timing per seed for \textsc{C2Greedy} is strictly decreasing over the size of $\mathcal{S}$, because the matrix $N$ shrinks, while per seed selection time of LT is increasing on average, because more seeds probably lead to bigger cascades.
\begin{figure}
\vskip -0.1in
\begin{center}
\subfigure[Total time]{\label{fig:bars}\includegraphics[scale=.21]{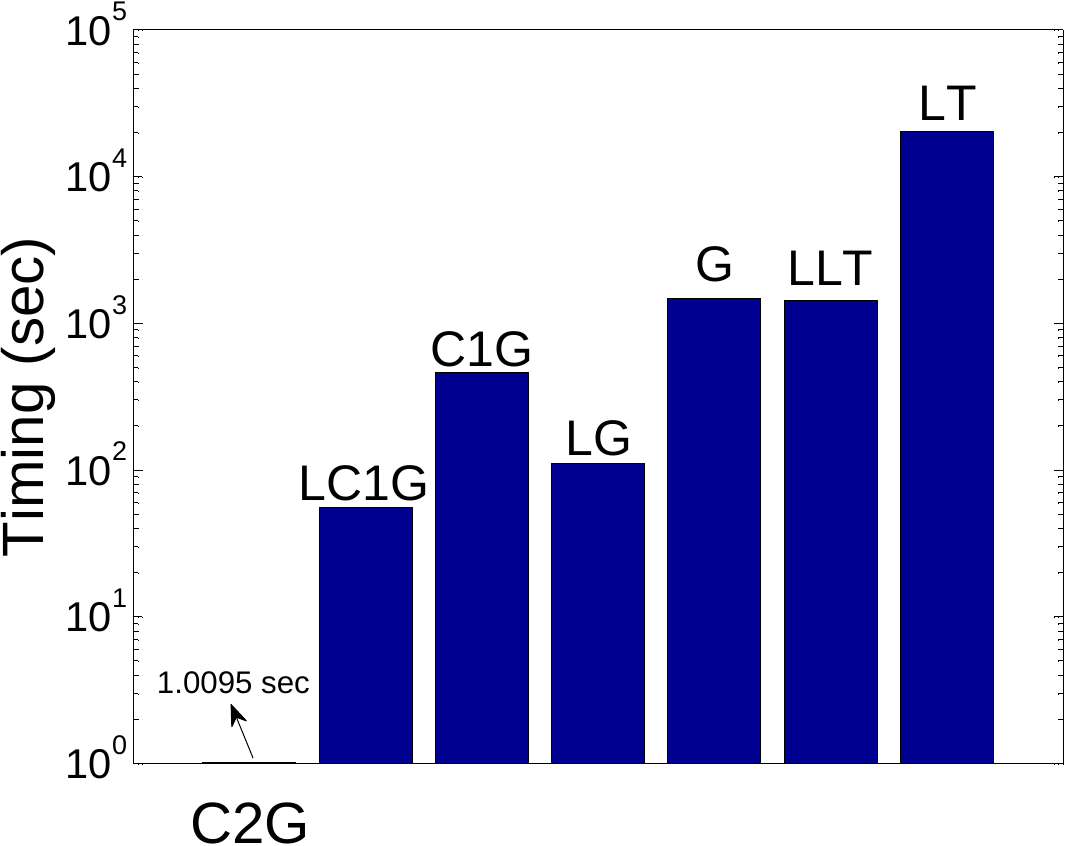}} ~
\subfigure[Time per seed 
(sec)]{\label{fig:amortized}\includegraphics[scale=.21]{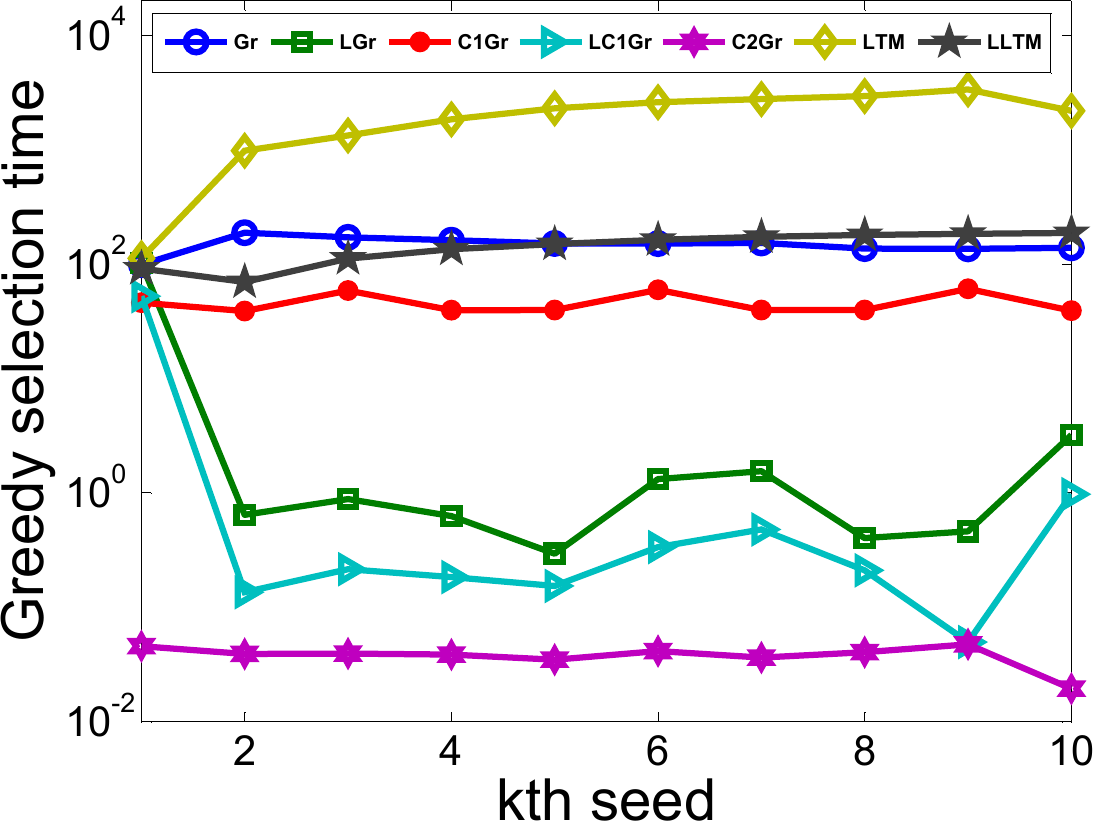}}
\vskip -0.1in
\caption{\small In (a) we compare the total timing of seven algorithms to investigate the effect of closed updates on speed and in (b) we show the per-seed required time for the same experiment.}
\label{fig:timing}
\end{center}
\vskip -0.15in
\end{figure} 

\begin{figure}
\vskip -0.1in
\begin{center}
\subfigure[Inverse Approximation]{\label{fig:invEst}\includegraphics[scale=.21]{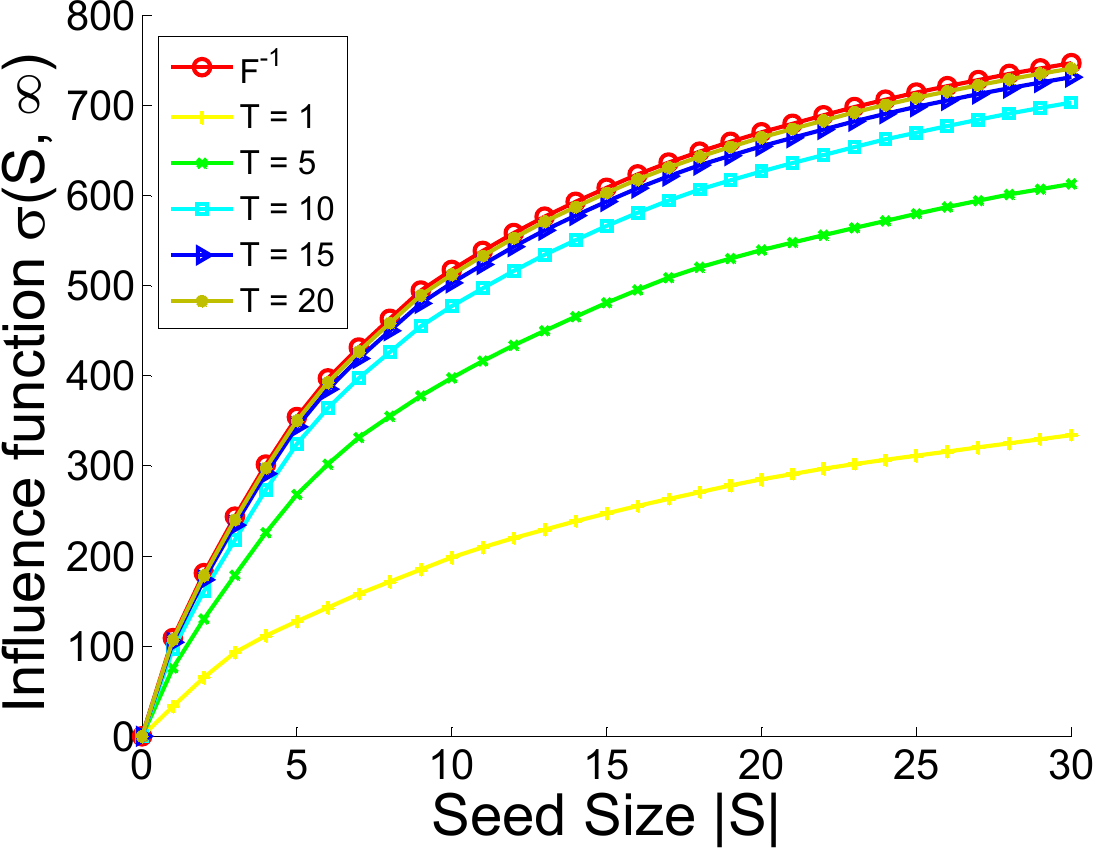}} ~
\subfigure[Scalability for $K=10$]{\label{fig:scalability}\includegraphics[scale=.21]{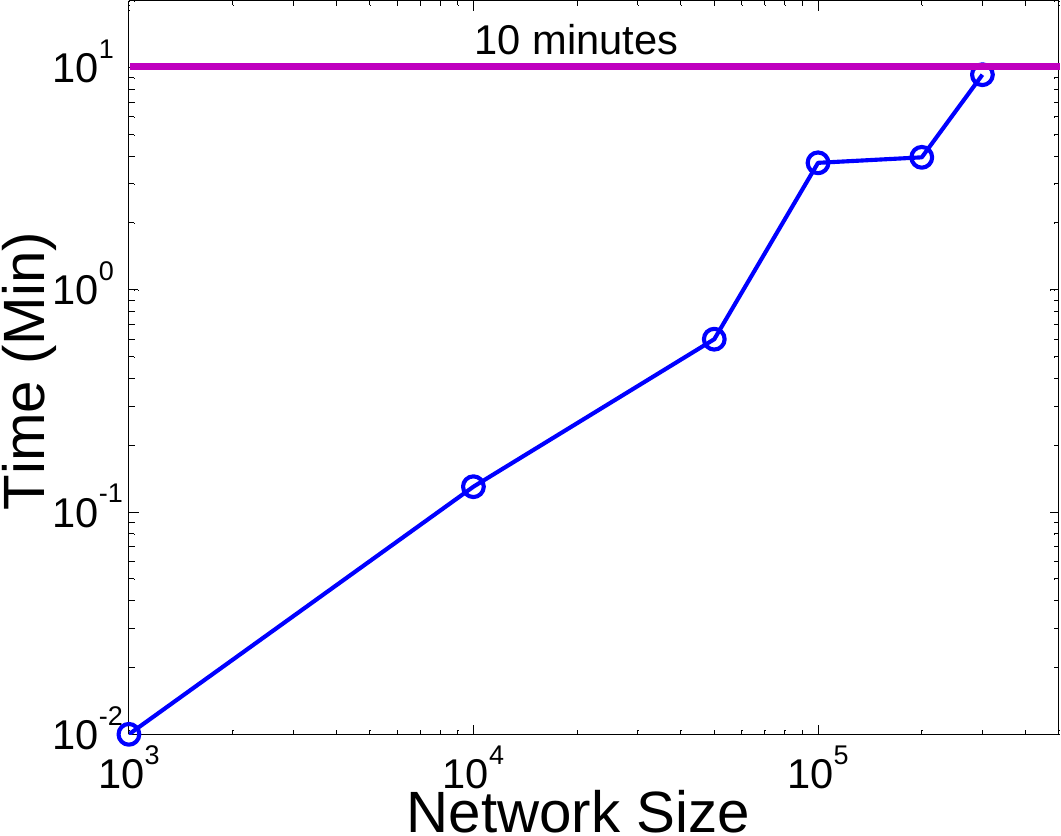}}
\vskip -0.1in
\caption{\small Timing for inf. max. in large scale networks by exploiting (a) inverse approximation and (b) parallel programming. Results of (b) are on FF networks with edge density 2.5. }
\end{center}
\vskip -0.15in
\end{figure} 

{\bf Inverse approximation.}
Going beyond networks of size $10^4$ makes the inverse computation problematic, but fortunately we have a good approximation of the inverse through the following expansion: $F=(I-R)^{-1}\approx I+R^1+R^2+...+R^T$. Since all eigenvalues of $R$ are less than or equal to 1 contribution of $(R)^i$ to the summation drops very fast as $i$ increases. 
The question is how many terms of the expansion, $T$, is enough for our application. 
Heuristically we choose the (effective) diameter of the graph as the number that provides us with a good approximation of $F^{-1}$. 
Note that the $i$th term of the expansion pertains to the shortest paths of size $i$ between any pair of nodes. 
Since the graph diameter is the longest shortest path between any pair of nodes, having that many terms gives us a good approximation of $F^{-1}$. 
This is also demonstrated by the experimental result of Figure \ref{fig:invEst} where we compare the result of the influence maximization on the WikiVote network with diameter 15, with actual $F^{-1}$ and its approximation for different $T$'s. 
As discussed when $T$ reaches to the diameter, the result of the algorithm that uses inverse approximation coincides with the algorithm that uses the exact inverse. 

{\bf Scalability.} Finally to show the scalability of \textsc{C2Greedy} we perform influence maximization on networks with sizes up to $3 \times 10^5$. For speeding up the large scale matrix computation of the Algorithm \ref{alg1} we developed an MPI version of our code which allows us to run \textsc{C2Greedy} on computing clusters. Figure \ref{fig:scalability} shows the running time of  \textsc{C2Greedy} for ForestFire networks of sizes varying between 1K to 300K with edge density 2.5 (i.e. ratio of edges to nodes) and effective diameter of 10. The MPI code was run on up to 400 cores of 2.8 GHz. As Figure \ref{fig:scalability} indicates even for the largest tested network with 0.3 million nodes and 0.75 million edges \textsc{C2Greedy} takes less than 10 minutes for $K=10$. 

To give a sense of our achievement in scalability we briefly mention the result of one of the state-of-the-art methods: The scalable ConTinEst \cite{due_scalable_2013} runs with 192 cores for almost 60 minutes on ForestFire network of size 100K and edge density of 1.5 to select 10 seeds, where our \textsc{C2Greedy} finishes in less than 2 minutes for the similar ForestFire network (100K nodes and density 1.8) with 200 cores.

\begin{figure}
\vskip -0.1in
\begin{center}
\subfigure[Non-progressive cascade of ML research topic.]{\label{fig:nprog}\includegraphics[scale=.21]{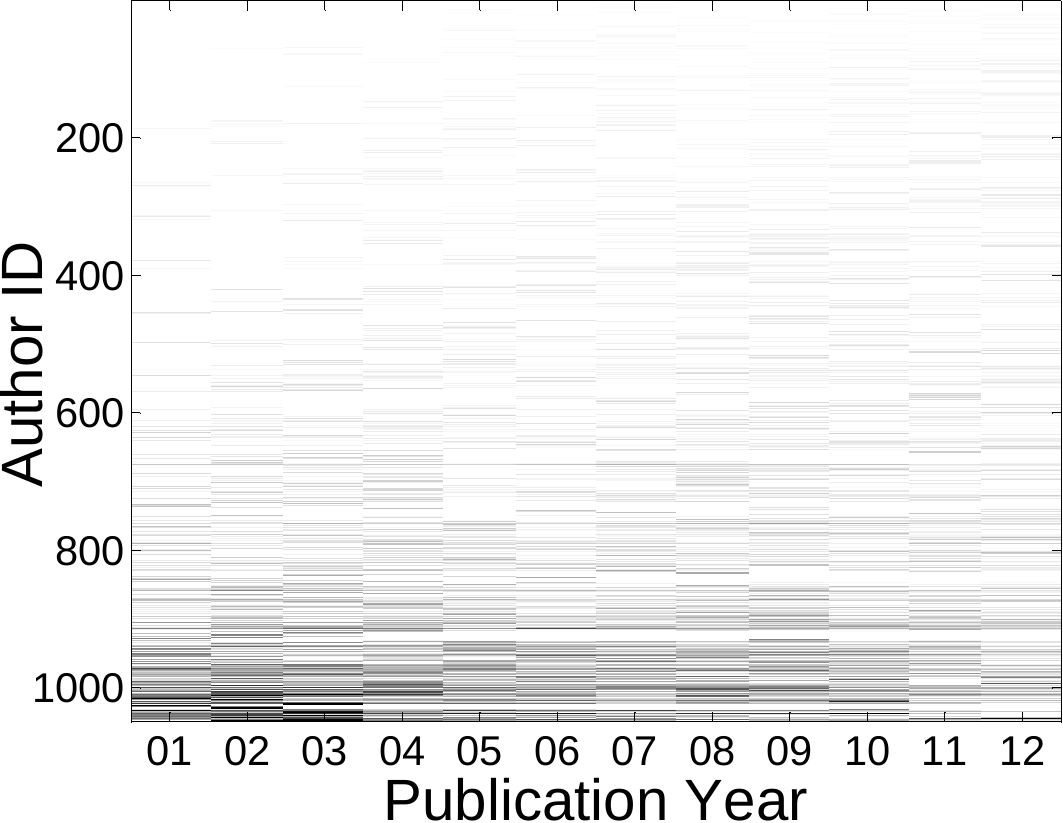}} ~
\subfigure[Inf. max. on inferred WFW network, MLWFW.]{\label{fig:realKmip}\includegraphics[scale=.21]{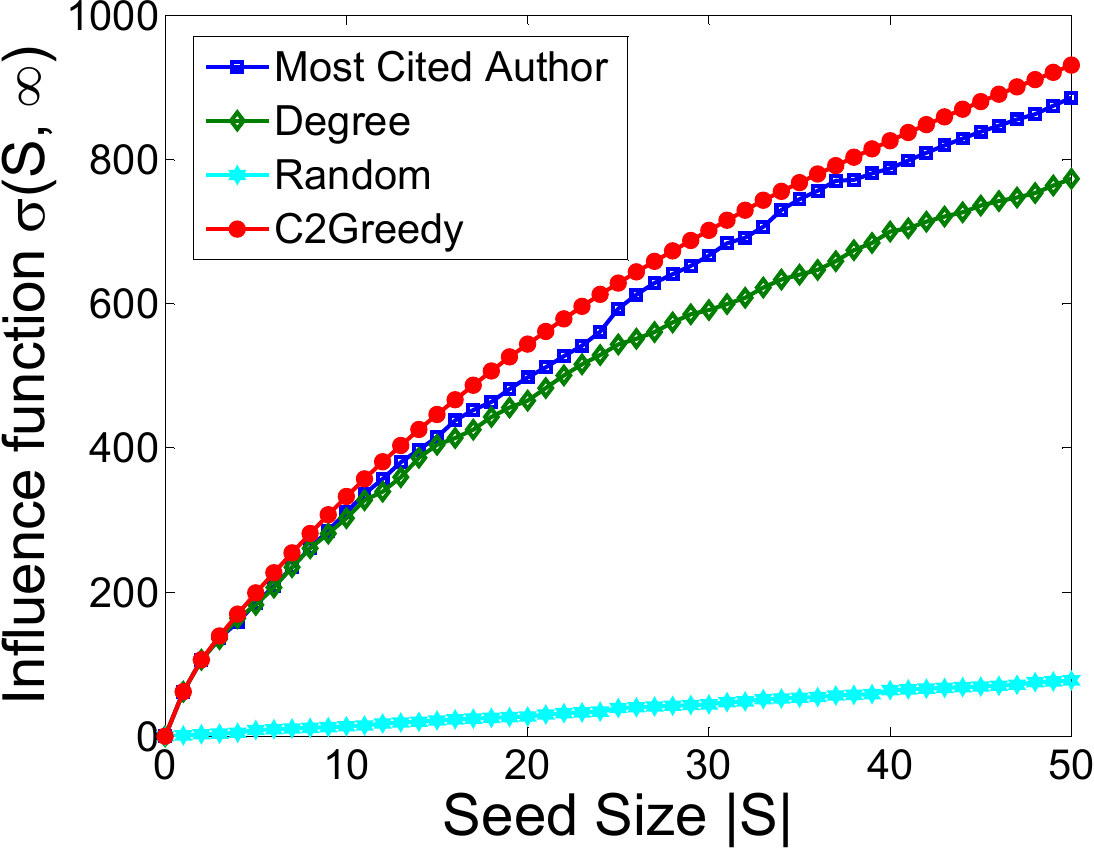}}
\vskip -0.1in
\caption{\small In (a) we show the existence of non-progressive cascade of ML research topic where white means all papers of the author is about ML. In (b) we compare C2Greedy result with other baselines such as most cited author.}
\label{fig:real}
\end{center}
\vskip -0.15in
\end{figure} 

\subsection{REAL NON-PROGRESSIVE CASCADE}\label{subsec:real}
\vspace*{-0.1in}
Collaboration and citation networks are two well-known real networks that have been studied in social network analysis literature \cite{kempe_maximizing_2003, tang_social_2007}. Here we introduce a new network that represents who-follows-whom (WFW) in a research community. 
Note that the nodes in the collaboration and citation networks are authors and papers respectively but in WFW network nodes are authors and edges are inferred from citations. 
A directed edges $(u,v)$ means that author $u$ has cited one of the papers of author $v$ which reveals that $u$ follows/reads papers of $v$. 
Here we investigate the ``research topic adoption'' cascade. 
Researchers adopt new research topics during their careers and influence their peers along different research communities. 
The process starts with an arbitrary research topic for each author and they are influenced by the research topic of those they follow and switch to another topic. 
For example a data mining researcher that follows mostly the papers of machine learning authors is probably going to switch his research topic to machine learning. 

For illustration, we consider only the authors who have published papers in Machine Learning (ML) conferences and journals in a given time period. 
For the list of ML related conferences and journal we use resources of ArnetMiner project \cite{tang_arnetminer:_2008}.
We consider each time step a year and study the years 2001 - 2012. 
An author is an \emph{active} ML author in a given year if at least half of his publications in that year was published in ML venues. 
Figure \ref{fig:nprog} shows the change in the percentage of ML publication of ML authors who has more than 70 publication in years between 2001 and 2012. 
As Figure \ref{fig:nprog} suggests, cascade of ML research topic is a non-progressive process and researcher switch back and forth between ML and other alternatives. Among 1049 authors of Figure \ref{fig:nprog} about 400 of them are core ML authors who have rarely published in any other topic, but the non-progressive nature of the process is more visible in the rest (bottom part of Figure \ref{fig:nprog}). 

Next we perform influence maximization on the inferred WFW network which we call MLWFW network. We extract the MLWFW network from the combined citation network of DBLP and ACM which is publicly available as a part of ArnetMiner project \cite{tang_arnetminer:_2008} and learn the edge weights similar to the weighted cascade model of \cite{kempe_maximizing_2003}. The MLWFW network of 2001 - 2012 time frame consists of 10604 authors and 168918 edges. Figure \ref{fig:realKmip} compares the result of influence maximization using \textsc{C2Greedy} and other baselines. Note that other than regular baselines in this specific domain we have another well-known method which is ``most cited author'' that is equal to selecting authors with highest weighted in-degree in MLWFW network. As Figure \ref{fig:realKmip} illustrates, \textsc{C2Greedy} outperforms all of the other methods. Note that the list of $K$ most influential authors in this experiment means that ``if'' those authors were switching to the ML topic completely (becoming a member of  seed set $\mathcal{S}$) they would make the topic vastly popular. Therefore, although the seed set contains the familiar names of well-known ML authors (e.g., Michael I. Jordan and John Lafferty in first and second places), sometimes we encounter exceptions. For example, in the list of top 10 authors selected by \textsc{C2Greedy} we have ``Emery N. Brown'' who is a renowned neuroscientist with publications in ``Neural Computation'' journal.

\vspace*{-2pt}
\section{CONCLUSION}
\vspace*{-0.05in}
We introduced the Heat Conduction Model which is able to capture both social influence and non-social influence, and extends many of the existing non-progressive models.  We also presented a scalable and provably near-optimal solution for influence maximization problem by establishing three essential properties of HC: 1) submodulairty of influence spread, 2) closed form computation for influence spread, and 3) closed form greedy selection. We conducted extensive experiments on networks with hundreds of thousands of nodes and close to million edges where our proposed method gets done in a few minutes, in sharp contrast with the existing methods. The experiments also certified that our method outperforms the state-of-the-art in terms of both influence spread and scalability. Moreover, we exhibited the first real non-progressive cascade dataset for influence maximization. We believe that our method removes the computational barrier that prevented the literature from considering the non-progressive influence models. 
Studying other forms of non-progressive influence models, such as non-progressive IC, is an interesting future work. 
 
\subsubsection*{\large{Acknowledgements}}
This research was supported in part by DTRA grants HDTRA1-09-1-0050 and HDTRA1-14-1-0040, DoD ARO
MURI Award W911NF-12-1-0385, and NSF grants CNS-10171647, CNS-1117536 and CNS-1411636.

\bibliographystyle{IEEEtran}
\bibliography{uai15}


\section{Supplementary}

\subsection{Influence maximization under progressive model: A brief review}
CELF method of Leskovec et al. \cite{leskovec_cost-effective_2007} attempts to speed up the original greedy method, proposed by Kempe et al. \cite{kempe_maximizing_2003}, by reducing the number of calls to Monte Carlo routine for spread computation. CELF lazy method is based on the submodularity of the influence spread and can be applied to any submodular maximization problem. Although lazy evaluation improves the running time of the original greedy method by up to 700 times \cite{leskovec_cost-effective_2007}, it still does not scale to large graphs \cite{chen_scalable_2010}.

Recently heuristics have been proposed to approximate influence spread for LT \cite{chen_scalable_2010} and IC \cite{chen_scalable_20101} which enables the greedy method to scale for large networks. Chen et al. \cite{chen_scalable_2010} suggest to use a local directed acyclic graph (LDAG) per node, instead of considering the whole graph, to approximate the influence flowing to the node. Goyal et al. propose SIMPATH method \cite{goyal2011simpath} under the LT model which is built on CELF method \cite{leskovec_cost-effective_2007}. They approximate the influence spread by enumerating the simple paths starting from the seeds within a small neighborhood. Both of these methods have parameters to be tuned which control the trade-off between running time and accuracy of influence spread estimation. Methods presented in \cite{chen_scalable_2010, goyal2011simpath} accelerate the greedy method \cite{kempe_maximizing_2003} substantially and achieve high performance in influence maximization.

Gomez-Rodriguez et al. \cite{rodriguez_uncovering_2011} propose a progressive continuous time influence model with dynamics similar to IC and show that influence maximization is NP-hard for this model as well. They show submodularity of influence spread and exploit the same greedy algorithm. In contrast to all other progressive models, influence spread has a closed form for this model but the computation is not scalable for large scale networks. A recent work \cite{due_scalable_2013} has scaled influence computation by developing a randomized algorithm for approximating it.

\subsection{Proof of Theorem 1}
For proving this theorem we need the following lemmas.
\begin{lemma}
When an interior node $s$ is added to the current absorbing set $\mathcal{S}$, the new fundamental matrix $F$ can be calculated from the previous one using the following equation:
\begin{equation}
F_{ij}^{\mathcal{S} \cup \{s\}}=F^{\mathcal{S}}_{ij}-\frac{F^{\mathcal{S}}_{is}F^{\mathcal{S}}_{sj}}{F^{\mathcal{S}}_{ss}}, \nonumber
\end{equation} 
\end{lemma}
\begin{proof}
The proof is straightforward based on Schur complement theorem \cite{zhang2006schur}.
\end{proof}
This lemma helps avoiding the matrix inversion required for computing the new $F^{\mathcal{S} \cup \{s\}}$ whenever an interior node $s$ is added to the seed set $\mathcal{S}$. 
\begin{lemma}
The expected number of passages through an interior node and the expected number of passages through its interior neighbors has the following relation:
\[
 F^{\mathcal{S}}_{ij} =
  \begin{cases}
   \sum_k F^{\mathcal{S}}_{ik}R^{\mathcal{S}}_{kj} & i \neq j \\
   1+\sum_k F^{\mathcal{S}}_{ik}R^{\mathcal{S}}_{kj} & i = j
  \end{cases}
\]
\end{lemma}
\begin{proof}
We know $F^S=(I-R^S)^{-1}$. Start with $(I-R^S)^{-1}(I-R^S)=I$ and after multiplication and rearranging we get to the lemma's statement: $F^S=I+F^SR^S$
\end{proof} 
\begin{lemma}
Starting from node $i$ the absorption probability by node $s$, when $S \cup \{s\}$ is the absorbing set, can be obtained from the expected number of passages through node $s$ when it was not absorbing:
\begin{equation}
Q_{is}^{\mathcal{S} \cup \{s\}}=\frac{F^{\mathcal{S}}_{is}}{F^{\mathcal{S}}_{ss}}.
\end{equation}
\end{lemma}
\begin{proof}
\begin{eqnarray}
Q_{is}^{\mathcal{S} \cup \{s\}} &=&\sum_{j\in \mathcal{V}\setminus\{\mathcal{S} \cup \{s\}\}} F^{\mathcal{S} \cup \{s\}}_{ij}B^{\mathcal{S} \cup \{s\}}_{js} \nonumber
\\&=&\sum_{j\in \mathcal{V}\setminus\{\mathcal{S}\}} F^{\mathcal{S} \cup \{s\}}_{ij}R^S_{js} \nonumber
\\&=&\sum_{j\in \mathcal{V}\setminus\{\mathcal{S}\}} (F^S_{ij}-\frac{F^S_{is} F^S_{sj}}{F^S_{ss}})R^S_{js} \nonumber
\\&=&\sum_{j\in \mathcal{V}\setminus\{\mathcal{S}\}} F^S_{ij}R^S_{js}-\frac{F^S_{is}}{F^S_{ss}}\sum_{j\in \mathcal{V}\setminus\{\mathcal{S}\}} F^S_{sj}R^S_{js} \nonumber
\\&=&F^S_{is}-\frac{F^S_{is}}{F^S_{ss}}(F^S_{ss}-1) \nonumber
\\&=&\frac{F^S_{is}}{F^S_{ss}}, \nonumber
\end{eqnarray}
where the third and fifth equalities come from lemma 1 and lemma 2 respectively.
\end{proof}

Proof of Theorem 1 is simply an instantiation of Lemma 3 for the case that we add node $s$ as the first seed to the network and get $Q_{is}^{\{s\}}=\frac{F^{\emptyset}_{is}}{F^{\emptyset}_{ss}}$, where $\emptyset$ emphasizes that the bias node is the only boundary.
Note that all of the three lemmas are general in a sense that absorbing set can contain any type of boundary points, including zero-value node like the bias node and one-value node like a seed node.

\subsection{Proof of Theorem 2}
\begin{proof}
Consider an instance of the NP-complete Vertex Cover problem
defined by an undirected and unweighted $n$-node graph $G = (\mathcal{V},\mathcal{E})$ and an integer
$k$; we want to know if there is a set $\mathcal{S}$ of $k$ nodes in $G$ so that
every edge has at least one endpoint in $\mathcal{S}$. We show that this can be
viewed as a special case of the influence maximization (\ref{eq:opt}).
Given an instance of the Vertex Cover problem involving a graph
$G$, we define a corresponding instance of the influence maximization
problem under HC for \textit{infinite time horizon}, by considering the following settings in (\ref{eq:binHC}): (i) $\omega_{ij}=\omega_{ji}=1$, if edge $(i-j)\in \mathcal{E}$, otherwise $\omega_{ij}=\omega_{ji}=0$, (ii) bias node's value is zero $b=0$, and (iii) $\beta_i$ for all $i$'s are equal to a known $\beta$. Note that since each interior node is connected to the zero-value bias node with edge weight $\beta$ it cannot have value larger than $1-\beta$. Hence, if there
is a vertex cover $\mathcal{S}$ of size $k$ in $G$, then one can deterministically
make $\sigma(\mathcal{A},\infty) = k+(n-k)(1-\beta)$ by targeting the nodes in the set $\mathcal{A} = \mathcal{S}$; conversely,
this is the only way to get a set $\mathcal{A}$ with $\sigma(\mathcal{A},\infty) = k+(n-k)(1-\beta)$. 
\end{proof}

\subsection{Proof of Theorem 3}
\label{subsec:proof}

As mentioned in Section \ref{sub:greedySel} when $t \to \infty$ superposition principle applies for HC model. We exploit this fact to prove the submodularity of influence spread. First note that $\sigma(\mathcal{S},\infty)$ computed from (\ref{eq:generalInf}) is the sum of node values and since the conic combination of submodular functions is also submodular it is enough to show that each node value, i.e., $v(i)$ is submodular to proof Theorem 3. 
Here we need to work with the general set of bias nodes (compare to single bias node $b$) which we call ground set $\mathcal{G}$. We introduce a new notation where the value of node $i$ is shown with $v^{\mathcal{S},\mathcal{G}}(i)$. Also seed nodes can have arbitrary value of $\geq b$ instead of all 1 values.For proving the submodularity of $v(i)$ we should prove: 
\begin{equation}
\label{subsig}
{v}^{\mathcal{T}\cup \{s\},\mathcal{G}}(i) -  {v}^{\mathcal{T},\mathcal{G}}(i)  
\geq 
{v}^{\mathcal{S}\cup \{s\},\mathcal{G}}(i) -  {v}^{\mathcal{S},\mathcal{G}}(i)
, \mathcal{T} \subseteq \mathcal{S} 
\end{equation}
We invoke superposition to perform the subtraction:
\begin{eqnarray}
\label{37}
{v}^{\{s_{v_L}\}, \mathcal{G} \cup \mathcal{T}}(i) 
\geq 
{v}^{\{s_{v_R}\}, \mathcal{G} \cup \mathcal{S}}(i) 
, \qquad \mathcal{T} \subseteq \mathcal{S}
\end{eqnarray}
where $v_L$ and $v_R$ emphasize that the value of the new seed node is different in left and right hand side and is qual to $v_L = \big(1 -  {v}^{\mathcal{T},\mathcal{G}}(s)\big)$ and $v_R = \big(1 -  {v}^{\mathcal{S},\mathcal{G}}(s)\big)$. Note that $v_L \geq v_R$ since $\mathcal{T} \subseteq \mathcal{S}$. 
We can not compare the value of nodes in two different networks unless they share same grounds and seeds with possibly different values for each seed. 
Therefore, we try to make the grounds of both sides of (\ref{37}) identical by expanding the LHS of (\ref{37}) using superposition law \cite{agarwal_foundations_2005}:
\begin{eqnarray}
\label{40}
{v}^{\{s_{v_L}\}, \mathcal{G} \cup \mathcal{T}}(i) 
= {v}^{\{s_{v_L}\}, \mathcal{G} \cup \mathcal{S}}(i) 
+ {v}^{\mathcal{D}, \mathcal{G} \cup \mathcal{S} \cup {s},}(i) \\ \nonumber
\end{eqnarray}
where $\mathcal{D} = \mathcal{S} - \mathcal{T}$. Although second term of (\ref{40}) is complicated but for our analysis it is enough to note that it is a non-negative number $\alpha \geq 0$. Now the submodularity inequality (\ref{subsig}) reduces to: 
\begin{equation}
{v}^{\{s_{v_L}\}, \mathcal{G} \cup \mathcal{S}}(i) + \alpha
\geq 
 {v}^{\{s_{v_R}\}, \mathcal{G} \cup \mathcal{T}}(i) 
\end{equation}
Now both sides have the same set of sources and grounds and we now $v_L(u) \geq v_R(u)$ and $\alpha \geq 0$ which completes the proof.

\end{document}